\documentclass[letterpaper, 10 pt, conference]{ieeeconf}
\IEEEoverridecommandlockouts
\usepackage{amsmath,amssymb}
\usepackage{mathtools}
\usepackage{bm}

\usepackage{amsthm} 
\usepackage{subfig}

\usepackage{xcolor}

\newtheorem{definition}{Definition}

\newtheorem{proposition}{Proposition}
\newtheorem{theorem}{Theorem}
\newtheorem{remark}{Remark}
\newtheorem{assumption}{Assumption}

\title{\LARGE \bf
Connectivity-Preserving Multi-Agent Area Coverage via Optimal-Transport-Based Density-Driven Optimal Control (D²OC)
}

\author{Kooktae Lee$^{1}$ and Ethan Brook$^{1}$
\thanks{This work was supported by NSF CAREER Grant CMMI-DCSD-2145810. $^{1}$Kooktae Lee and Ethan Brook are with the Department of Mechanical Engineering, New Mexico Institute of Mining and Technology, Socorro, NM 87801, USA, email: kooktae.lee@nmt.edu, ethan.brook@student.nmt.edu.}  
}

\begin{document}

\maketitle
\thispagestyle{empty}
\pagestyle{empty}

\begin{abstract}
Multi-agent systems are widely used for area coverage tasks in applications such as search-and-rescue, environmental monitoring, and precision agriculture. Achieving \emph{non-uniform coverage}, where certain regions are prioritized, requires coordinating agents while accounting for dynamic and communication constraints. Existing density-driven methods effectively distribute agents according to a reference density but typically do not guarantee connectivity, which can lead to disconnected agents and degraded coverage in practical deployments. This letter presents a connectivity-preserving approach within the Density-Driven Optimal Control (D$^2$OC) framework. The coverage problem, expressed via the Wasserstein distance between agent distributions and a reference density, is formulated as a quadratic program. Communication constraints are incorporated through a smooth penalty function, ensuring strict convexity and global optimality while naturally maintaining inter-agent connectivity without rigid formations. Simulation results demonstrate that the proposed method effectively keeps agents within communication range, improving coverage quality and convergence speed compared to methods without explicit connectivity enforcement.
\end{abstract}

\section{Introduction}
Multi-agent systems have increasingly attracted attention for area coverage tasks, including search-and-rescue, environmental monitoring, precision agriculture, and infrastructure inspection \cite{oh2015survey,cortes2004coverage,mathew2011metrics}. Coordinated agents can explore large or complex environments more efficiently than a single agent by distributing sensing and actuation, reducing mission time and energy consumption. In many applications, some regions require greater attention due to environmental significance, hazard likelihood, or mission priorities, necessitating \emph{non-uniform coverage} strategies \cite{lee2022density}.

\textbf{Literature Survey:}  
Various approaches have been proposed for non-uniform coverage. Classical density-driven methods, such as Heat Equation Driven Area Coverage (HEDAC) \cite{ivic2023multi} and Density-Driven Control (D$^2$C) \cite{lee2022density}, operate in a decentralized manner: each agent computes its motion locally based on a reference density. In contrast, Spectral Multiscale Coverage (SMC) \cite{mathew2009spectral} is centralized, requiring global knowledge of all agents and the reference distribution. While decentralization provides flexibility and adaptability, agents must still communicate to maintain coordinated behavior. Existing methods generally do not explicitly enforce connectivity, which can lead to disconnected agents and degraded coverage.

\textbf{Connectivity Preservation:}  
Maintaining inter-agent communication is critical for decentralized coordination. Conventional strategies rely on fixed formations or rigid inter-agent distance constraints \cite{oh2015survey, zhao2021formation, afrazi2025density}, which reduce flexibility in cluttered or dynamic environments. Opportunistic communication \cite{mox2024opportunistic} allows agents to exchange information as opportunities arise but can slow convergence and cause uneven coverage.  
Recent connectivity-preserving MPC methods \cite{carron2023multi, kawajiri2021coverage} maintain communication via algebraic-connectivity or distance-based penalties, but typically require centralized or iterative computation and do not address non-uniform coverage.  
In this study, a connectivity-preserving mechanism is embedded \emph{within} a convex density-driven optimal control framework, ensuring scalability, decentralization, and robust communication.

\textbf{Density-Driven Optimal Control (D$^2$OC):}  
This work builds on the D$^2$OC framework~\cite{seo2025density}, where a reference density defines spatial priorities and the Wasserstein distance measures how well the agent distribution aligns with this reference. Control inputs steer agents to reduce this discrepancy while respecting dynamics and motion constraints. Unlike~\cite{seo2025density}, which derived an optimal-control solution using Lagrange multipliers, the present work reformulates the problem to enable a quadratic-program (QP) representation and further incorporates connectivity-preserving constraints to maintain inter-agent communication during coverage.

\textbf{Contribution:}  
This work enhances D$^2$OC by enforcing connectivity during non-uniform coverage. The main contributions are summarized as follows:  
\textbf{1) QP Equivalence and Convexity Analysis:} The Wasserstein-based D$^2$OC objective, \textit{without connectivity constraints}, is shown to be equivalent to a QP, and its strict convexity and associated optimal solution, including the closed-form unconstrained optimizer, are established. These results were not presented in~\cite{seo2025density};  
\textbf{2) Connectivity-Preserving Penalty with Reachable Sets:} A smooth connectivity penalty integrated with a reachable-set formulation is developed to maintain inter-agent communication without enforcing rigid formations. When this penalty is included, the overall formulation remains convex but is no longer quadratic, an aspect absent from~\cite{seo2025density};  
\textbf{3) Simulation-Based Validation:} Representative simulations confirm that the proposed connectivity-preserving D$^2$OC achieves improved coverage efficiency, faster convergence, and sustained communication compared with the unconstrained case.

\section{Problem Setup}
\noindent \textbf{Notation:} 
The sets of real and integer numbers are denoted by \(\mathbb{R}\) and \(\mathbb{Z}\), respectively. The sets \(\mathbb{Z}_{>0}\) and \(\mathbb{Z}_{\geq 0} := \mathbb{Z}_{>0} \cup \{0\}\) denote positive and non-negative integers.  
The space \(\mathbb{R}^n\) represents \(n\)-dimensional column vectors. 
The Euclidean and infinity norms are denoted by $\|\cdot\|_2$ (simply $\|\cdot\|$ when obvious) and $\|\cdot\|_\infty$, respectively,
and the transpose of a matrix $A$ by $A^\top$.
The zero matrix $\mathbf{0}_{m\times n} \in \mathbb{R}^{m\times n}$ and the identity matrix $\mathbf{I}_n \in \mathbb{R}^{n\times n}$ are denoted with their sizes as subscripts. 
A weighted norm is defined as \(\|U\|_R := \sqrt{U^\top R U}\), where \(R \succ 0\).  
The operator \(\mathrm{diag}([\cdot])\) constructs a block diagonal matrix from its arguments. 
The operator \(\mathrm{blkdiag}(\cdot)\) denotes a block-diagonal concatenation, i.e., \(\mathrm{blkdiag}(A_h)_{h=r}^{r+H-1}\) places each block \(A_h\) on the diagonal.  
The Hadamard (elementwise) product is denoted by \(\odot\), and 
\(\oplus\) denotes the Minkowski sum, i.e., 
\(A \oplus B = \{a+b \mid a\in A,\, b\in B\}\).

We study a network of agents, each described by discrete-time linear dynamics. For agent \(i\) in the multi-agent system, the evolution over time index \(k \in \mathbb{Z}_{\ge 0}\) is modeled by the Linear Time-Invariant (LTI) system as
\begin{equation}
\begin{aligned}
x_i(k+1) &= A_i x_i(k) + B_i u_i(k), \quad
y_i(k) = C_i x_i(k),
\end{aligned}\label{eq:dyn}
\end{equation}
where \(x_i(k) \in \mathbb{R}^n\) denotes the state vector, \(u_i(k) \in \mathbb{R}^m\) the control action, and \(y_i(k) \in \mathbb{R}^d\) the measured output. The system, input, and output matrices are given by \(A_i \in \mathbb{R}^{n \times n}\), \(B_i \in \mathbb{R}^{n \times m}\), and \(C_i \in \mathbb{R}^{d \times n}\). 

To guarantee that the reachable sets we later use remain meaningful and well-behaved, the following standard assumptions are adopted.

\begin{assumption}\label{assump:controllability}
For each agent \(i\), the pair \((A_i, B_i)\) is controllable.
\end{assumption}
\vspace{-.15in}
\begin{assumption}\label{assump:stability}
For each agent \(i\), the state matrix \(A_i\) is marginally stable: all eigenvalues lie in the closed unit disk, and any eigenvalue on the unit circle has equal algebraic and geometric multiplicity.
\end{assumption}
\vspace{-.15in}
\begin{assumption}
\label{assump:neighbor_model}
Each agent is assumed to know the nominal dynamics $(A_j,B_j,C_j)$ of those agents with which communication connectivity is to be preserved, as these models are specified during system integration. At each control step, agents update the most recent neighbor output (or position) information received.
\end{assumption}

These assumptions are essential for the communication-aware coverage problem. Assumption \ref{assump:controllability} ensures that each agent can maneuver its own state through admissible inputs, while Assumption \ref{assump:stability} guarantees bounded state evolution. Assumption \ref{assump:neighbor_model} allows an agent to compute the reachable sets of its neighbors from their current states and known dynamics, which is necessary for enforcing communication constraints over the prediction horizon.

\subsection{Wasserstein Distance and Optimal Transport}

To formalize the notion of non-uniform coverage, we employ optimal transport theory \cite{villani2008optimal}, with particular emphasis on the Wasserstein distance. For two discrete probability measures \(\rho\) and \(\nu\) on a metric space \((\mathcal{X},d)\), the \(p\)-Wasserstein distance is given by
\begin{equation}
\textstyle
\mathcal{W}_p(\rho,\nu) = \left( \min_{\pi_{\ell j}} \sum_{\ell=1}^M \sum_{j=1}^N \pi_{\ell j}\, d(y_\ell,q_j)^p \right)^{1/p},\label{eq:W-dist}
\end{equation}
subject to the constraints
\[
\textstyle
\pi_{\ell j}\geq 0, \, \sum_{j=1}^N \pi_{\ell j} = \alpha_\ell, \, \sum_{\ell=1}^M \pi_{\ell j} = \beta_j, \, \sum_{\ell,j}\pi_{\ell j} = 1,
\]
where \(\pi_{\ell j}\) specifies the amount of probability mass transported from \(y_\ell\) to \(q_j\). In this paper we focus on the quadratic case (\(p=2\)) with Euclidean distance as the ground cost for $d(\cdot)$.

We distinguish between two categories of points.  
\emph{Agent points} \(y_\ell(k) \in \mathbb{R}^d\) denote the positions of agents at time step \(k\), evolving according to the LTI dynamics \eqref{eq:dyn}.  
\emph{Sample points} \(q_j \in \mathbb{R}^d\) represent fixed reference locations that encode the desired spatial distribution.  

Each agent has a finite operation time and produces at most \(M_i\) agent points, yielding a total of \(M=\sum_{i=1}^{n_a}M_i\) points for total \(n_a\) agents. For simplicity, uniform weights are assumed: \(\alpha_\ell = 1/M_i\) for agent points and \(\beta_j=1/N\) for samples.

To assess coverage quality, we compare the empirical agent distribution with the reference distribution:
\begin{equation}
\rho(k) = \tfrac{1}{k+1}\sum_{t=0}^k\bigg(\tfrac{1}{n_a}\sum_{i=1}^{n_a}\delta_{y_i(t)}\bigg), 
\, 
\nu = \frac{1}{N}\sum_{j=1}^N \delta_{q_j},
\end{equation}
where \(\delta_y\) is the Dirac measure. The discrepancy between the two is quantified at time \(k\) by \(\mathcal{W}_2(\rho(k),\nu)\). The objective of Density-Driven Optimal Control (D$^2$OC) is to minimize the Wasserstein distance \(\mathcal{W}_2(\rho(k),\nu)\) between the empirical agent distribution \(\rho(k)\) and the reference density \(\nu\), subject to constraints such as the number of agents, operation time, and communication range.

\subsection{Decentralized Coverage Protocol}
Directly minimizing \(\mathcal{W}_2(\rho(k),\nu)\) is challenging due to its nonconvexity and high dimensionality. To address this, each agent solves a local subproblem that considers only nearby sample points rather than the full reference map. By defining a local Wasserstein distance over these points, agents can compute feasible control inputs that progressively reduce the overall discrepancy while satisfying dynamics and actuation constraints.  

Within this framework, D$^2$OC organizes agent behavior into three recurring stages:

\begingroup
\setlength{\leftmargini}{12pt} 
\begin{enumerate}
    \item \textbf{Sample selection and control input:}  
    Each agent identifies nearby sample points with relatively high remaining weights and computes a feasible control input to reduce its local Wasserstein distance while respecting dynamics and actuation limits.
    
    \item \textbf{Weight adjustment:}  
    After executing its control input, agent $i$ updates (reduces) the weights of the sample points $j$ it has covered or influenced, denoted by $\beta_{i,j}(k)$, recording its coverage progress at time $k$.
    
    \item \textbf{Information exchange for multi-agent collaboration:}  
    When agents come within communication range, they synchronize weight information by adopting the minimum observed weights among neighbors. This ensures global coverage consistency and prevents redundant exploration.
\end{enumerate}
\endgroup

The first two stages are executed independently by each agent (decentralized control), while the third stage enables coordination through weight sharing among neighbors. Repeated execution of this cycle gradually aligns the empirical distribution of agents with the reference distribution over their operational horizon. The objective of this study is to determine the optimal control input under a connectivity-preserving constraint in the first stage, while the full description of each stage can be found in \cite{seo2025density}.

\section{Formulation of the D$^2$OC Cost Function via Quadratic Programming}

This section presents the formulation of the D$^2$OC optimization problem in terms of the Wasserstein distance. In general, for a discrete-time system, the control input $u(k)$ does not immediately influence the system output $y(k)$ at the same time step. Instead, the input affects the output after a certain number of steps, which is formalized using the concept of output relative degree.

\begin{definition}[Output Relative Degree of a Discrete-Time LTI System]\label{def:output_relative_degree}
Consider the discrete-time LTI system \eqref{eq:dyn}. The \emph{output relative degree} \(r \in \mathbb{Z}_{>0}\) is the smallest positive integer such that
$C_i A_i^{r-1} B_i \neq 0$, and $C_i A_i^{\ell-1} B_i
= 0$ for all $\ell = 1, \ldots, r-1$.
\end{definition}

This defines the number of time steps required for the control input \(u_i(k)\) to first have a direct effect on the output vector \(y_i(k)\). Then, the cost function for agent $i$ using the squared local Wasserstein distance to achieve D$^2$OC is defined over the prediction horizon \(H \in \mathbb{Z}\), starting from time \(k+r\):
{\small
\begin{equation}
\sum_{h=r}^{r+H-1} \mathcal{W}^2_i(k+h) := \sum_{h=r}^{r+H-1} \sum_{j \in \mathcal{S}_i({k+h)}} \pi_j(k+h) \, \|y_i(k+h) - q_j\|^2,
\label{eq:local_wasserstein_cost}
\end{equation}
}
subject to agent dynamics \eqref{eq:dyn} and the Wasserstein distance constraints in \eqref{eq:W-dist}.
The symbol $\mathcal{S}_i(k+h)$ denotes the set of \emph{local sample points} selected for agent $i$ at time $k+h$, based on (i) remaining weights $\beta_{i,j}(k+h)$ of sample points, prioritizing points not yet fully covered, and (ii) proximity to agent $i$, ensuring computational tractability and focus on nearby regions (see \cite{seo2025density}). The transport weight $\pi_j(k+h)$ represents agent $i$’s contribution to sample $q_j$ under the local optimal transport plan. The agent index does not appear in $\pi_j$ since it corresponds to a single point. This formulation captures the cost of moving agents to assigned targets while respecting dynamics over the prediction horizon.

Leveraging this property, we establish the following result.

\medskip
\begin{proposition}\label{proposition:equiv}
Let $\mathcal{S}_i(k+h)$ denote the index set of local sample points for agent $i$ 
at time $k+h$. 
Let the transport weights $\pi_j(k+h) \ge 0$ be 
locally computed over the prediction window
$h = r,\ldots,r+H{-}1$ based on the selected local samples.
Define the weighted barycenter at \(k+h\) as
$
\bar{q}_i(k+h) := \dfrac{1}{\sum_{j \in \mathcal{S}_i(k+h)} \pi_j(k+h)}\sum_{j \in \mathcal{S}_i(k+h)} \pi_j(k+h) q_j
$
and 

\begin{equation}
\begin{aligned}
Y_i^{k|r:H} &:= [\, y_i(k+r)^\top \!~\cdots~\! y_i(k+r\!+\!H\!-\!1)^\top \,]^\top,\\
\bar Q_i^{k|r:H} &:= [\, \bar q_i(k+r)^\top \!~\cdots~\! \bar q_i(k+r\!+\!H\!-\!1)^\top \,]^\top,\\
\boldsymbol{\Omega}_i^{k|r:H} &:= 
\mathrm{blkdiag}\!\big(
\sqrt{\textstyle\sum_{j\in\mathcal{S}_i(k+h)}\!\pi_j(k+h)}\,\mathbf I_d
\big)_{h=r}^{r+H-1}.
\end{aligned}\label{eq:Y,Q,Omega}
\end{equation}

With `const.' denoting all terms independent of $Y_i^{k|r:H}$, we have
{\small
\begin{align*}
\sum_{h=r}^{r+H-1} \mathcal{W}^2_i(k+h)
&= \left\| \boldsymbol{\Omega}_i^{k|r:H} \big( Y_i^{k|r:H} - \bar Q_i^{k|r:H} \big) \right\|^2 + \mathrm{const.},
\end{align*}
}
\end{proposition}

\begin{proof}
Expanding the quadratic Wasserstein term for each $h$,
{\small
\begin{align*}
&\sum_{h=r}^{r\texttt{+}H-1}\!\mathcal{W}^2_i(k\texttt{+}h)
= \sum_{h=r}^{r\texttt{+}H-1}\!\sum_{j\in\mathcal{S}_i(k\texttt{+}h)} 
\pi_j(k\texttt{+}h)\|y_i(k\texttt{+}h)-q_j\|^2\\
&=\sum_{h=r}^{r\texttt{+}H-1}\!
\Big(\!\sum_j \pi_j(k\texttt{+}h)\!\Big)\!
\|y_i(k\texttt{+}h)-\bar{q}_i(k\texttt{+}h)\|^2
\texttt{+}\!\sum_{h=r}^{r\texttt{+}H-1}\!C(h),
\end{align*}
}
where \(C(h) := \sum_{j} \pi_j(k+h) \| q_j - \bar{q}_i(k+h) \|^2\) is independent of the decision variables. Stacking the terms yields the compact form
$
\big\| \boldsymbol{\Omega}_i^{k|r:H}( Y_i^{k|r:H} - \bar{Q}_i^{k|r:H}) \big\|^2 + \mathrm{const.}
$
\end{proof}

To apply Proposition~1 within the optimal control formulation, 
the stacked output $Y_i^{k|r:H}$ in \eqref{eq:Y,Q,Omega} over the horizon is written in affine form using the stacked input
$U_i^{k|H} := [\,u_i(k)^\top ~ \cdots ~ u_i(k+H-1)^\top\,]^\top \in \mathbb{R}^{mH}$ as
\begin{equation}
    Y_i^{k|r:H} = \Theta_i U_i^{k|H} + \Phi_i x_i(k),
    \label{eq:Y_i}
\end{equation}
where $\Theta_i \in \mathbb{R}^{dH\times mH}$ and 
$\Phi_i \in \mathbb{R}^{dH\times n}$ are given by
{\small
\begin{align}
\Theta_i &:=
\begin{bmatrix}
C_i A_i^{r-1} B_i & \mathbf{0} & \cdots & \mathbf{0} \\
C_i A_i^{r}   B_i & C_i A_i^{r-1} B_i & \cdots & \mathbf{0} \\
\vdots & \vdots & \ddots & \vdots \\
C_i A_i^{r+H-2} B_i & C_i A_i^{r+H-3} B_i & \cdots & C_i A_i^{r-1} B_i
\end{bmatrix},\label{eq:Theta}\\[1ex]
\Phi_i &:=
\begin{bmatrix}
(C_i A_i^{r})^\top,\,
(C_i A_i^{r+1})^\top,\,
\ldots,\,
(C_i A_i^{r+H-1})^\top
\end{bmatrix}^{\top}.
\label{eq:Phi}
\end{align}
}

Combining \eqref{eq:local_wasserstein_cost} with the input penalty $\|U_i^{k|H}\|^2_{R_i}$, where $R_i \succ 0$, yields
\begin{align}
J(U_i^{k|H}) := \sum_{h=r}^{r+H-1} \mathcal{W}^2_i(k+h) + \|U_i^{k|H}\|^2_{R_i}.
\label{eq:D2OC_optimization_prob}
\end{align}

Finally, by utilizing Proposition~1 together with the input--output relation \eqref{eq:Y_i}, the control objective becomes the quadratic form
\begin{equation}
\begin{aligned}
J(U_i^{k|H})
&= \tfrac12 (U_i^{k|H})^\top H_i U_i^{k|H}
   + f_i^\top U_i^{k|H} + \mathrm{const.}, \\
H_i &:= 2\!\left((\boldsymbol{\Omega}_i^{k|r:H}\Theta_i)^\top
                 (\boldsymbol{\Omega}_i^{k|r:H}\Theta_i) + R_i\right),\\
f_i &:= 2(\boldsymbol{\Omega}_i^{k|r:H}\Theta_i)^\top
        \boldsymbol{\Omega}_i^{k|r:H}(\Phi_i x_i(k)-\bar Q_i^{k|r:H}),
\label{eq:QP}
\end{aligned}
\end{equation}
where `const.' collects all terms independent of $U_i^{k|H}$.

\begin{theorem}[Uniqueness of the Unconstrained Optimal Input]
\label{theorem:W_min_opt_con}
For agent \(i\) in the multi-agent system, governed by the LTI dynamics
\eqref{eq:dyn} with output relative degree \(r\) and prediction horizon \(H\),
the quadratic cost in \eqref{eq:QP} has the unconstrained optimal input
\begin{equation}
(U_i^{k|H})^{\mathrm{uncon}} = -\,H_i^{-1} f_i, \label{eq:unc_opt}
\end{equation}
and this solution is the unique global minimizer.
\end{theorem}

\begin{proof}
Taking the gradient of \eqref{eq:QP} and setting it to zero yields
\(
H_i U_i^{k|H} + f_i = 0
\),
so the unconstrained minimizer is
\(
(U_i^{k|H})^{\mathrm{uncon}} = -H_i^{-1} f_i.
\)

Since \(R_i \succ 0\), the Hessian 
\(
H_i = 2\big((\boldsymbol{\Omega}_i^{k|r:H}\Theta_i)^{\!\top}
         (\boldsymbol{\Omega}_i^{k|r:H}\Theta_i) + R_i\big)
\)
satisfies \(H_i \succ 0\).
Thus the cost is strictly convex and the minimizer is unique.
\end{proof}

\section{Connectivity-Preserving D$^2$OC}

\subsection{Connectivity Constraint with Reachable Sets}

To maintain communication over the prediction horizon, agent $i$ and a designated neighbor $j$ satisfy
\begin{equation}
\begin{aligned}
g_{ij}(k+h) := &~\|x_i(k+h) - x_j(k+h)\|^2 \\
&- r_{\mathrm{comm}}^2 \le 0, \quad h = r,\dots,r+H-1,
\end{aligned}
\end{equation}
with relative degree $r$, horizon $H$, and communication radius $r_{\mathrm{comm}}>0$.  
The neighbor $j$ is selected according to a user-specified connected communication topology (e.g., chain, tree), and global connectivity is preserved as long as the resulting graph remains connected.
Since agent $i$ cannot know agent $j$'s exact future outputs, we describe agent $j$'s possible outputs via a reachable set.

\subsubsection*{Reachable-Set Formulation}
Let the reachable set of agent $j$ in output space be the zonotope
\begin{equation}
\mathcal{Z}^y_j(k+h) = \hat{y}_j(k+h) \oplus G_j(k+h)\, \mathrm{diag}(\Delta u_j)\, \mathcal{B}_\infty,\label{eq:r-set}
\end{equation}
where $\hat{y}_j(k+h)$ denotes the nominal prediction based on $(A_j,B_j,C_j)$ and the most recently exchanged output. Here, $\mathcal{B}_\infty = \{ z \in \mathbb{R}^m : \|z\|_\infty \le 1 \}$ is the unit hypercube scaled by the half-range $\Delta u_j = (u_{\max}-u_{\min})/2$, where $u_{\max}$ and $u_{\min}$ are the maximum and minimum admissible inputs. Note that the input shift $\hat{u}_j = (u_{\max}+u_{\min})/2$ is not included in \eqref{eq:r-set}, which is required for the exact reachable set, while it is unnecessary for computing the conservative radius of the reachable set.
The generator matrix $G_j(k{+}h) \in \mathbb{R}^{d\times m}$ maps bounded control deviations to output deviations:
\begin{equation}
G_j(k{+}h)
=
\big[
C_j A_j^{\,h-1} B_j,\;
C_j A_j^{\,h-2} B_j,\;
\ldots,\;
C_j A_j^{\,r-1} B_j
\big],
\label{eq:G_def}
\end{equation}
with $G_j(k{+}h)=0$ for $h<r$. Each column represents the effect of one input direction.

Equivalently,
\[
\mathcal{Z}^y_j(k+h) = \{ \hat{y}_j(k+h) + G_j(k+h)\, \mathrm{diag}(\Delta u_j)\, z \mid \|z\|_\infty \le 1 \}.
\]

\subsubsection*{Conservative Scalar Inequality}
For tractability, approximate the set inclusion with
\[
R_j(k+h) = \max_{\|z\|_\infty \le 1} \| G_j(k+h)\, \mathrm{diag}(\Delta u_j)\, z \|,
\]
and enforce
$\|y_i(k+h) - \hat{y}_j(k+h)\| \le r_{\mathrm{comm}} - R_j(k+h)$, $h=r,\dots,r+H-1$,
ensuring that the reachable set lies within agent $i$'s communication ball.

For fast online checks, a conservative bound is
\[
R_j(k{+}h) = \max_{\ell} \| G_{j,\ell}(k{+}h)\, \Delta u_j \|,
\]
where $G_{j,\ell}(k{+}h)$ is the $\ell$-th column of $G_j(k+h)$.

\subsection{Soft-Constraint for Connectivity-Preserving Control}
We relax the feasibility requirement by introducing a \emph{soft penalty} on communication range violations.  
Specifically, for each horizon step $h \in \{r,\dots,r+H-1\}$, let the predicted output be
\begin{equation}
y_i(k+h) = C_i A_i^h x_i(k) + \sum_{s=0}^{h-1} C_i A_i^{h-1-s} B_i\, u_i(k+s).
\end{equation}

Defining
$f_i(h) = C_i A_i^h x_i(k) - \hat y_j(k+h)$,
$F_i(h) = \begin{bmatrix} C_i A_i^{h-1} B_i  \cdots & C_i B_i \end{bmatrix}$,
the nominal inter-agent distance becomes
\begin{equation}
d_i(h) := \| y_i(k+h) - \hat{y}_j(k+h) \| = \| f_i(h) + F_i(h) U_i^{k|h} \|.\label{eq:d_i}
\end{equation}

To penalize violations of the communication threshold $r_{\mathrm{comm}} - R_j(k+h)$, 
We now introduce the smooth penalty using the log-sum-exp
\begin{align}
\phi(T_i(h)) &= \frac{\kappa}{\eta} \log \big( 1 + \exp(\eta \, T_i(h)) \big), 
\quad \kappa, \eta > 0,
\end{align}
where $T_i(h) := d_i(h) - \big(r_{\mathrm{comm}} - R_j(k+h)\big)$ and $\kappa>0$ scales the penalty and $\eta>0$ controls its steepness.

This convex function vanishes when $d_i(h) \le r_{\mathrm{comm}} - R_j(k+h)$ 
and increases smoothly when the threshold is exceeded.  
The overall soft-constraint cost now becomes
\begin{equation}
J_{\text{soft}}(U_i^{k|H}) = J(U_i^{k|H}) + \textstyle\sum_{h=r}^{r+H-1} \phi(T_i(h)). \label{eq:J_soft}
\end{equation}

\begin{remark}[Extension to Collision Avoidance]
The formulation in \eqref{eq:J_soft} can be further extended to include an additional soft constraint for inter-agent collision avoidance by penalizing violations of a minimum distance $d_{\min}$.
\end{remark}

\begin{remark}[Robustness to Communication Delays and Uncertainties]
The reachable-set formulation provides inherent tolerance to limited communication delays and imperfect knowledge of neighboring agents' behavior. When output updates from neighbors are delayed, their effects are captured by the conservative zonotope radius \(R_j(k+h)\), ensuring that connectivity is maintained. The soft-constraint term further improves robustness by smoothly penalizing temporary violations instead of enforcing strict feasibility.
\end{remark}

\begin{theorem}[Strict Convexity of Soft-Constraint Control with Box Constraints]
\label{thm:soft_box}
Consider the discrete-time linearized dynamics of agent $i$ in \eqref{eq:dyn} with output relative degree $r$.  
Let the soft-constrained D$^2$OC cost be given by \eqref{eq:J_soft}, subject to the admissible input set
\[
\mathcal U_i^{k|H} := \{ U_i^{k|H} \in \mathbb{R}^{mH} \mid u_{\min}^{(H)} \le U_i^{k|H} \le u_{\max}^{(H)} \}.
\]
Then, $J_{\mathrm{soft}}$ is strictly convex, and the optimization problem admits a unique minimizer.
\end{theorem}

\begin{proof}
The proof follows the original convexity argument, with the box constraints included.
For each \(h\), the predicted output \(y_i(k+h)\) is an affine function of the 
stacked input \(U_i^{k|H}\), and the inter-agent distance is given by \eqref{eq:d_i}.

\paragraph{Convexity of the distance term}
The map $U_i^{k|H} \mapsto d_i(h)$ is the Euclidean norm composed with an affine map, hence convex.

\paragraph{Convexity and monotonicity of the scalar penalty}
Consider the scaled log-sum-exp function
\begin{equation}
\phi(z) = \frac{\kappa}{\eta} \log\big(1 + \exp(\eta z)\big), 
\quad \kappa, \eta > 0.
\end{equation}
Its derivatives are
$
\phi'(z) = \kappa \dfrac{\exp(\eta z)}{1 + \exp(\eta z)} \ge 0,\,\, 
\phi''(z) = \kappa \dfrac{\eta \, \exp(\eta z)}{(1 + \exp(\eta z))^2} \ge 0,
$
thus $\phi$ is convex and nondecreasing.

\paragraph{Composition with shifted distance}
Define the shifted distance
\begin{equation}
\textstyle
T_i(h) := d_i(h) - \big(r_{\mathrm{comm}} - R_j(k+h)\big),
\end{equation}
which is convex in $U_i^{k|h}$. By the composition rule for convex functions (convex, nondecreasing scalar composed with convex vector-valued map), $\phi(T_i(h))$ is convex in $U_i^{k|h}$.

\paragraph{Sum of convex terms and box constraints}
The total objective is
\begin{equation}
\small
\begin{aligned}
J_{\mathrm{soft}}(U_i^{k|H}) 
= &\frac12 (U_i^{k|H})^\top H_i U_i^{k|H} + f_i^\top U_i^{k|H} + \sum_{h=r}^{r+H-1} \phi(T_i(h)).
\end{aligned}
\end{equation}
Each term is convex; finite sums of convex functions are convex. The feasible set 
\(\mathcal U_i^{k|H} = \{ U_i^{k|H} \in \mathbb{R}^{mH} \mid u_{\min}^{(H)} \le U_i^{k|H} \le u_{\max}^{(H)} \}\)
is closed and convex. Hence, $\min_{U_i^{k|H} \in \mathcal U_i^{k|H}} J_{\mathrm{soft}}(U_i^{k|H})$ is convex.

\paragraph{Existence of minimizers}
If $\mathcal U_i^{k|H}$ is bounded, continuity of $J_{\mathrm{soft}}$ guarantees existence. 
If $\mathcal U_i^{k|H}$ is unbounded but $H_i \succ 0$, coercivity of the quadratic term ensures existence. 

\paragraph{Uniqueness under $H_i \succ 0$}
The quadratic term $\frac12 (U_i^{k|H})^\top H_i U_i^{k|H}$ is strictly convex. Adding convex penalties preserves strict convexity. Restricting a strictly convex function to the convex set $\mathcal U_i^{k|H}$ yields a unique minimizer.

\paragraph{KKT characterization}
Let $\lambda^\pm \ge 0$ be Lagrange multipliers for the box bounds. At the minimizer $(U_i^{k|H})^\star$,
\begin{equation*}
    \begin{aligned}
&\nabla J_{\mathrm{soft}}((U_i^{k|H})^\star) + \lambda^+ - \lambda^- = 0, \qquad
\lambda^\pm \ge 0,\\
&\lambda^+ \odot ((U_i^{k|H})^\star - u_{\max}^{(H)}) = 0,\,
\lambda^- \odot (u_{\min}^{(H)} - (U_i^{k|H})^\star) = 0.
    \end{aligned}
\end{equation*}

\medskip
Combining these steps, the soft-constrained problem with box constraints is strictly convex, and a minimizer exists and is unique since $H_i \succ 0$.
\end{proof}

\begin{remark}[Soft Connectivity and Effect of Penalty Parameters]\label{remark:soft and tuning}
The soft-constraint formulation provides a convex relaxation of the hard connectivity condition, allowing limited violations while maintaining overall communication. The penalty parameters $\kappa$ and $\eta$ jointly determine the strength and sharpness of this relaxation. Increasing $\kappa$ reinforces the
penalty and improves connectivity preservation, whereas larger $\eta$ produces a steeper transition toward a hard constraint and may introduce numerical stiffness if too large. In practice, $\kappa$ is increased until violations are consistently penalized, and $\eta$ is then selected as the smallest value that
sharpens the penalty without causing oscillatory or saturated behavior. Smaller $\eta$ values yield smoother but weaker enforcement, whereas larger $\eta$ values impose stricter connectivity at the cost of reduced smoothness. A margin $\gamma \in (0,1)$ scales the communication range $r_{\mathrm{comm}}$ to reduce soft-constraint violations.
\end{remark}

\begin{remark}[Computational Scalability]
In the decentralized D$^2$OC framework, each agent solves an independent convex optimization problem \eqref{eq:J_soft} using local information, enabling parallel computation and maintaining scalability as the number of agents increases. 
The prediction horizon affects only the dimension of each local problem and remains tractable for typical receding-horizon settings.
\end{remark}

\section{Simulations}
To evaluate the proposed connectivity-preserving D$^2$OC, simulations were performed with a twenty-agent system using a full 12-state linearized quadrotor model~\cite{powers2015quadrotor}, which has an output relative degree of $r=4$. The reference map was generated from a 3D point cloud distribution, shown as green dots in Fig.~\ref{fig:sim}(a,c), where the agents' final positions are marked with yellow circles. All agents were initialized near the origin (red crosses) and clustered at the start. 
With connectivity constraints, a chain-like topology was used, where 
agent~\(i\) remained within the communication range of agent~\(i{+}1\) 
for \(i<20\). The prediction horizon was set to \(H{=}1\) for simplicity. A larger horizon would provide better foresight in coordination but at 
a higher computational cost. In this setting, each local optimization 
required about \(10\) ms per agent in MATLAB, indicating that real-time 
execution is readily achievable in faster compiled implementations.
\begin{table}[h!]
\centering
\small
\caption{Key Simulation Parameters}
\label{tab:sim_params_compact}
\setlength{\tabcolsep}{6pt}
\renewcommand{\arraystretch}{0.9}
\begin{tabular}{l c}
\hline
\textbf{Description} & \textbf{Value} \\
\hline
Number of agents & $20$ \\
Number of reference sample points & $500$\\
Sampling period $T$ & $0.1$ s \\
Prediction horizon $H$ / relative degree $r$ & $1$ / $4$ \\
Communication range threshold $r_{\mathrm{comm}}$ & $15$ \\
Connectivity margin $\gamma$ & $0.8$ \\
Connectivity soft penalty $(\kappa,\eta)$ & $(750,\,0.25)$ \\
Minimum inter-agent distance $d_{\min}$ & $1$ \\
Agent maximum speed $v_{\max}$ & $10$ \\
\hline
\end{tabular}
\end{table}

\begin{figure}[tbph!]
    \centering
\subfloat[Trajectories without constraint]{
    \includegraphics[width=0.53\linewidth]{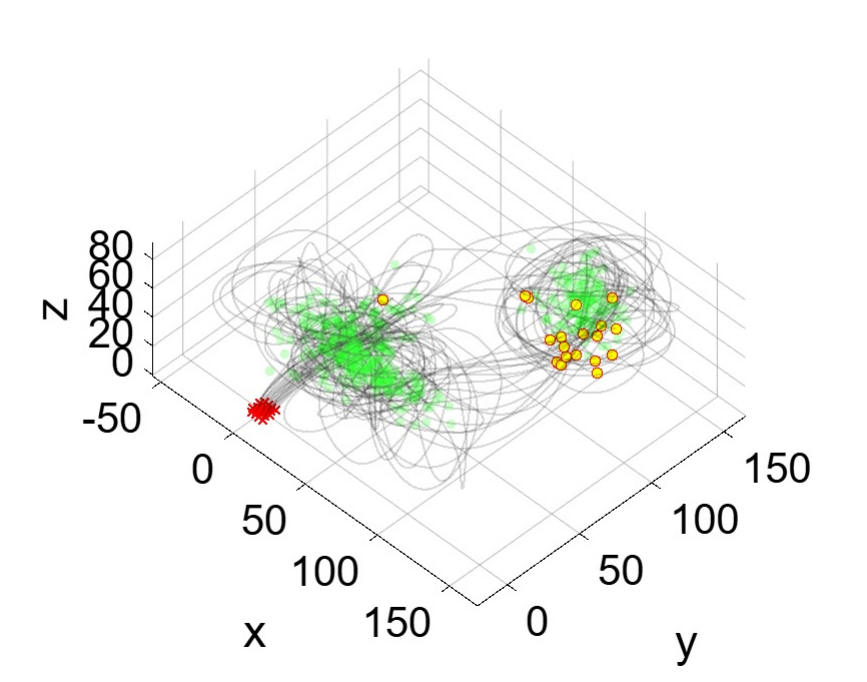}}\,
\subfloat[Inter-agent distances without constraint]{
    \includegraphics[width=0.42\linewidth]{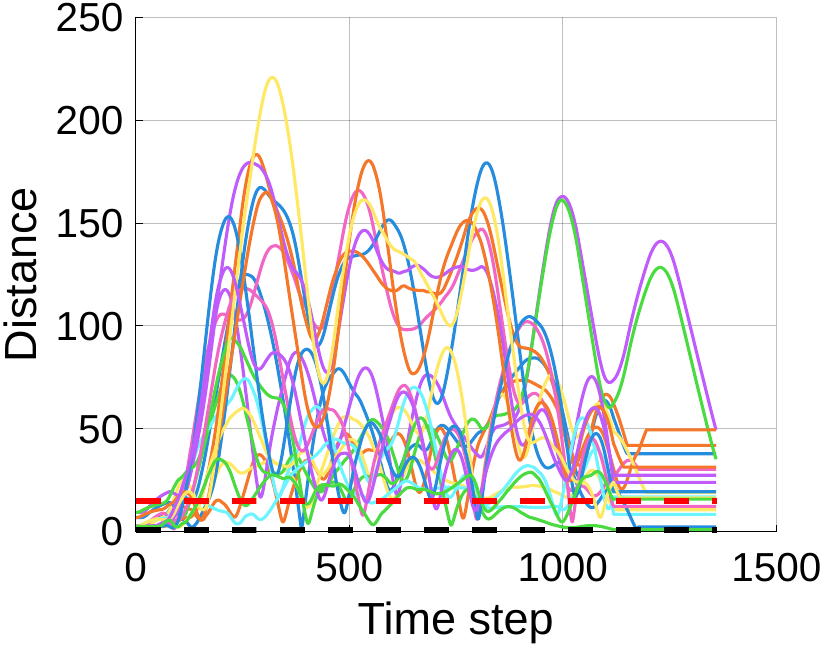}}\\
\subfloat[Trajectories with $r_{\mathrm{comm}}=15$]{
    \includegraphics[width=0.53\linewidth]{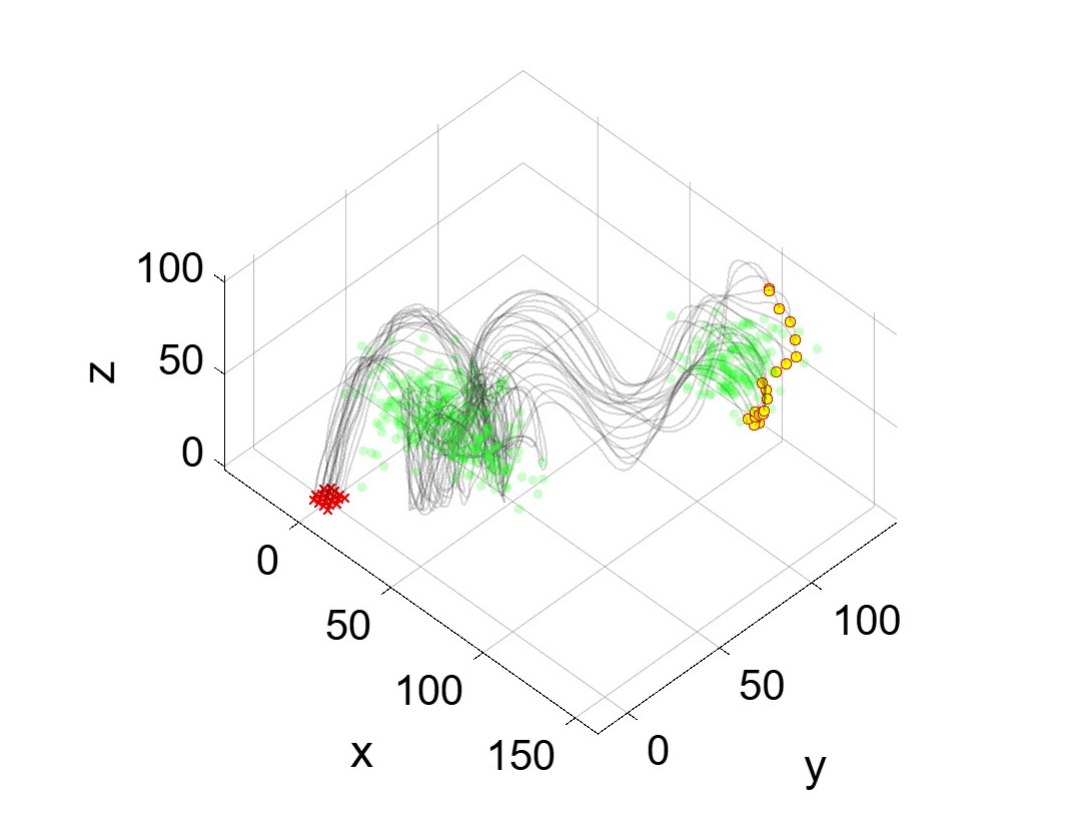}}\,
\subfloat[Inter-agent distances with $r_{\mathrm{comm}}=15$]{
    \includegraphics[width=0.42\linewidth]{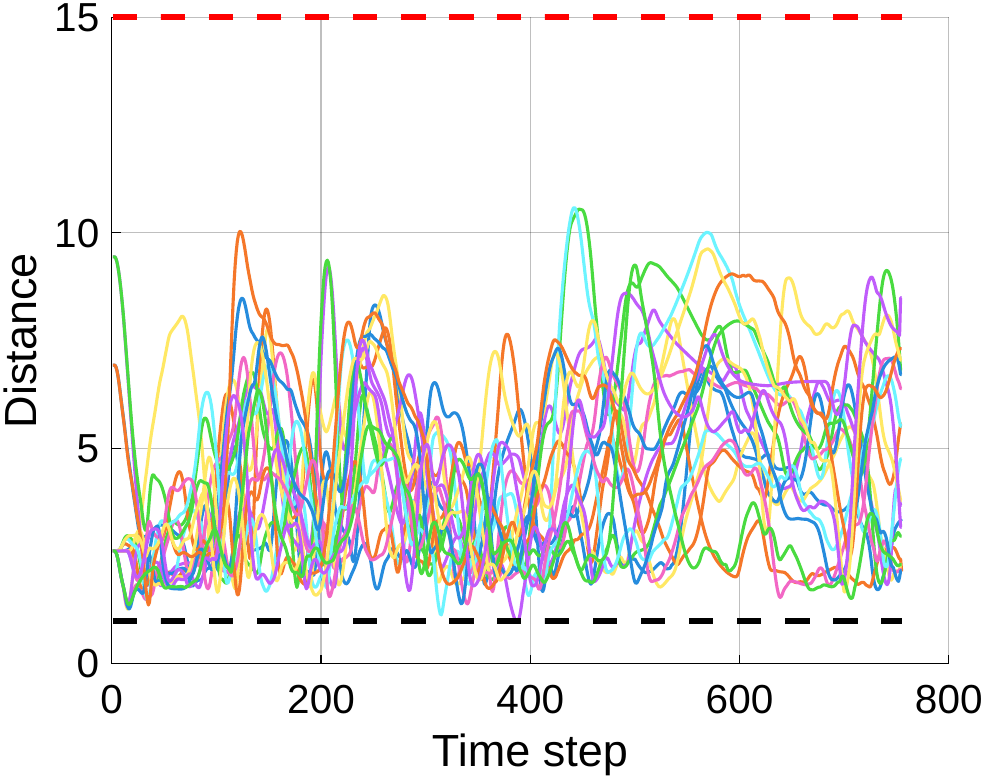}}
\caption{Twenty-agent D$^2$OC simulation: trajectories (left) and inter-agent distances (right) without/with connectivity constraints. Red dashed: communication threshold; black dashed: minimum distance.}
    \label{fig:sim}
\vspace{-1em}
\end{figure}

Two scenarios were tested: (i) no connectivity constraint and (ii) connectivity constraint with $r_{\mathrm{comm}}=15$.
Trajectories were computed in \textsc{MATLAB}, using \textsc{quadprog} for the unconstrained case and \textsc{fmincon} for the constrained cases. 
In all scenarios, box input constraints were imposed to satisfy the small-angle condition of the linearized model. 
The right column of Fig.~\ref{fig:sim} shows the corresponding inter-agent distances, verifying whether connectivity was preserved.

Without the connectivity constraint (Figs.~\ref{fig:sim}a,b), agents moved freely and dispersed widely across the domain. 
Although the communication range was \(r_{\mathrm{comm}} = 15\), the absence of constraint allowed inter-agent distances to reach as high as 200, as shown in Fig.~\ref{fig:sim}(b). 
Trajectories appeared to cover the reference samples (green) uniformly. However, since connectivity was not enforced, agents could share weight information, which represents local coverage, only when they temporarily came within range of others, resulting in event-based communication. 
Once disconnected, they lost awareness of covered regions, causing redundant revisits and inefficient exploration. 

When the connectivity-preserving constraint was applied (Figs.~\ref{fig:sim}c,d), agents maintained cohesive motion through relay links. 
Although collision avoidance is not the main focus of this work, a minimum distance of~1 was additionally enforced via a second soft constraint. 
While individual freedom was slightly reduced, continuous communication enabled synchronized weight updates and improved overall mission efficiency. 
The connectivity-preserving penalty was implemented with parameters $\kappa=750$, $\eta=0.25$, and $\gamma=0.8$ (Remark~\ref{remark:soft and tuning}), ensuring communication maintenance, whereas the second soft constraint independently kept safe separation of at least~1. 

For quantitative evaluation, the sliced Wasserstein distance (SWD) was used to measure how closely the distribution formed by agent trajectories matched the reference distribution, where a smaller value indicates better alignment. 
The unconstrained case yielded an SWD of 82.54 compared with 63.72 for the constrained case and required 1,374 versus 756 iterations for completion. 
These results demonstrate that the proposed connectivity-preserving mechanism enables faster and more accurate coverage with stronger spatial consistency among agents.

\section{Conclusion}

This letter presented a connectivity-preserving approach for multi-agent non-uniform area coverage within the Density-Driven Optimal Control (D$^2$OC) framework. By formulating the coverage problem via the Wasserstein distance as a quadratic program and incorporating communication constraints through a smooth penalty function, the proposed method ensures that agents remain within communication range while achieving effective coverage. The resulting formulation is strictly convex, allowing for globally optimal control inputs without imposing rigid formations. Simulation studies demonstrated that the approach improves both coverage quality and convergence speed compared to methods without explicit connectivity enforcement. 

Future work includes extending the framework to dynamic environments with 
time-varying communication and sensing conditions, and incorporating emergency 
protocols for restoring connectivity when links are lost. We also plan to 
conduct experimental validation of the proposed connectivity-preserving 
D$^2$OC framework using a multi-agent platform.



\bibliographystyle{ieeetr}
\bibliography{references}

\begin{thebibliography}{10}

\bibitem{oh2015survey}
K.-K. Oh, M.-H. Park, and H.-S. Ahn, ``A survey of multi-agent formation
  control,'' {\em Automatica}, vol.~53, pp.~424--440, 2015.

\bibitem{cortes2004coverage}
J.~Cortes, S.~Martinez, T.~Karatas, and F.~Bullo, ``Coverage control for mobile
  sensing networks,'' {\em IEEE Transactions on Robotics and Automation},
  vol.~20, no.~2, pp.~243--255, 2004.

\bibitem{mathew2011metrics}
G.~Mathew and I.~Mezi{\'c}, ``Metrics for ergodicity and design of ergodic
  dynamics for multi-agent systems,'' {\em Physica D: Nonlinear Phenomena},
  vol.~240, no.~4-5, pp.~432--442, 2011.

\bibitem{lee2022density}
K.~Lee and R.~Hasan~Kabir, ``Density-aware decentralised multi-agent
  exploration with energy constraint based on optimal transport theory,'' {\em
  International Journal of Systems Science}, vol.~53, no.~4, pp.~851--869,
  2022.

\bibitem{ivic2023multi}
S.~Ivi{\'c}, B.~Crnkovi{\'c}, L.~Grb{\v{c}}i{\'c}, and L.~Matlekovi{\'c},
  ``Multi-uav trajectory planning for 3d visual inspection of complex
  structures,'' {\em Automation in Construction}, vol.~147, p.~104709, 2023.

\bibitem{mathew2009spectral}
G.~Mathew and I.~Mezic, ``Spectral multiscale coverage: A uniform coverage
  algorithm for mobile sensor networks,'' in {\em Proceedings of the 48h IEEE
  Conference on Decision and Control (CDC) held jointly with 2009 28th Chinese
  Control Conference}, pp.~7872--7877, IEEE, 2009.

\bibitem{zhao2021formation}
Y.~Zhao, Y.~Hao, Q.~Wang, Q.~Wang, and G.~Chen, ``Formation of multi-agent
  systems with desired orientation: a distance-based control approach,'' {\em
  Nonlinear Dynamics}, vol.~106, no.~4, pp.~3351--3361, 2021.

\bibitem{afrazi2025density}
M.~Afrazi, S.~Seo, and K.~Lee, ``Density-driven formation control of a
  multi-agent system with an application to search-and-rescue missions,'' in
  {\em 2025 American Control Conference (ACC)}, pp.~3622--3627, IEEE, 2025.

\bibitem{mox2024opportunistic}
D.~Mox, K.~Garg, A.~Ribeiro, and V.~Kumar, ``Opportunistic communication in
  robot teams,'' in {\em 2024 IEEE International Conference on Robotics and
  Automation (ICRA)}, pp.~12090--12096, IEEE, 2024.

\bibitem{carron2023multi}
A.~Carron, D.~Saccani, L.~Fagiano, and M.~N. Zeilinger, ``Multi-agent
  distributed model predictive control with connectivity constraint,'' {\em
  IFAC-PapersOnLine}, vol.~56, no.~2, pp.~3806--3811, 2023.

\bibitem{kawajiri2021coverage}
S.~Kawajiri, K.~Hirashima, and M.~Shiraishi, ``Coverage control under
  connectivity constraints,'' in {\em Proceedings of the 20th International
  Conference on Autonomous Agents and MultiAgent Systems}, pp.~1554--1556,
  2021.

\bibitem{seo2025density}
S.~Seo and K.~Lee, ``Density-driven optimal control for efficient and
  collaborative multiagent nonuniform coverage,'' {\em IEEE Transactions on
  Systems, Man, and Cybernetics: Systems}, vol.~55, no.~12, pp.~9340--9354,
  2025.

\bibitem{villani2008optimal}
C.~Villani, {\em Optimal transport: old and new}, vol.~338.
\newblock Springer Science \& Business Media, 2008.

\bibitem{powers2015quadrotor}
C.~Powers, D.~Mellinger, and V.~Kumar, ``Quadrotor kinematics and dynamics,''
  {\em Handbook of Unmanned Aerial Vehicles}, pp.~307--328, 2015.

\end{thebibliography}

\end{document}